\documentclass[11pt,a4paper]{article}
\usepackage{latexsym,amsthm,amssymb}

\newtheorem{theorem}{Theorem}

\newtheorem{proposition}[theorem]{Proposition}

\newtheorem{lemma}[theorem]{Lemma}
\newcommand{\be}{\begin{equation}}
\newcommand{\ee}{\end{equation}}
\newcommand{\bea}{\begin{eqnarray}}
\newcommand{\eea}{\end{eqnarray}}
\newcommand{\ba}{\begin{array}}
\newcommand{\ea}{\end{array}}
\newcommand{\bean}{\begin{eqnarray*}}
\newcommand{\eean}{\end{eqnarray*}}

\newcommand{\Ga}{\Gamma}

\newcommand{\la}{\lambda}

\newcommand{\om}{\omega}

\newcommand{\de}{\delta}

\newcommand{\pa}{\partial}

\newcommand{\resi}{{\rm res\/}}

\begin{document}
\title
     {\bf On the String Equation of the BKP Hierarchy\/}
\author
{ Hsin-Fu Shen$^1$ and Ming-Hsien Tu$^2$
 \footnote{E-mail: phymhtu@ccu.edu.tw} \\ \\
  $^1$
  {\it Department of Mechanical Engineering, WuFeng Institute of Technology,\/}\\
  {\it Chiayi 621, Taiwan\/},\\
  $^2$
  {\it Department of Physics, National Chung Cheng University,\/}\\
  {\it  Chiayi 621, Taiwan\/}\/}
\date{\today}
  \maketitle
  \begin{abstract}
The Adler-Shiota-van Moerbeke formula is employed to derive the
$W$-constraints for the $p$-reduced BKP hierarchy constrained by
 the string equation. We also provide the Grassmannian description of the
 string equation in terms of the spectral parameter.
       \\ \\
PACS: 02.30.Ik\\ \\
Keywords:  BKP hierarchy, additional symmetries, vertex operators,
string equation, $W$-constraints, Sato Grassmannian.
\end{abstract}

\newpage
\section{Introduction}
String equation is an important constrained condition connecting
integrable hierarchies with solvable string theories and intersection theory
(see e.g. \cite{Dij91,vM94} and references therein). The most famous example is the KdV
hierarchy constrained by a string equation whose solutions corresponds to
partition function of 2d quantum gravity or generating function of
intersection number\cite{W90,W91,Kont92}(see also \cite{IZ}).
Since the solutions of KdV hierarchy can be
characterized by a single function called $\tau$-function, it turns out that
the constrained KdV hierarchy can be written as  Virasoro constraints on the $\tau$-function,
 in which  the lowest one corresponds to the string equation\cite{DVV,G91,L91a,L91b}.
 For higher order KdV hierarchy
 (i.e. $p$-reduced KP hierarchy with $p>2$) the associated $\tau$-function satisfies the
 $W$-constraints so that the $\tau$-function is a null-vector of a set of $W_p^+$-algebra
 which contains the Virasoro algebra as a subalgebra with central charge $c=p$
 \cite{FKN92,AM92,vM94}.
It was pointed out\cite{L91a,L91b,D93} that the string equation associated with the KP/KdV hierarchy
can be traced back to additional symmetries (or non-isospectral flows )of the hierarchies
proposed by Orlov and Schulman\cite{OS86}(see also \cite{D03}).
Furthermore, Adler-Shiota-van Moerbeke \cite{ASM94,ASM95} and Dickey\cite{D95} showed that
additional symmetries acting on wave function are connected
with Sato's B\"acklund symmetries acting on tau function.
These properties provide a useful tool to study the solutions to the string equation
of the $p$-reduced KP hierarchy with $p>2$\cite{AM92} that generalize Kontsevich's result
for the KdV case.

 In the present work, we like to study the string equation associated with the
 BKP hierarchy\cite{DJKM,T93}. The BKP hierarchy possesses many integrable structures
 as the KP hierarchy such as Lax formulation, tau function, Hirota bilinear equation,
 fermion representation, soliton and quasi-periodic solutions, etc.
 In \cite{L95} van de Leur provided the Adler-Shiota-van Moerbeke  formula for the BKP hierarchy.
 He also studied the $W$-constraint for the $p$-reduced BKP hierarchy in terms of
 the twisted affine Lie algebra $\hat{sl}_n^{(2}$\cite{L96}.
 In this work, we follow closely the work by Adler-van Moerbeke\cite{AM92} and Dickey\cite{D93}
 to reconsider the string equation of the BKP hierarchy from the point of view of
  additional symmetries in Lax-Orlov-Schulman formulation.
  In particular, we shall show that the invariance of the
$p$-reduced BKP hierarchy  with respect to a particular additional symmetry is crucial for
obtaining the string equation which together with the  Adler-Shiota-van Moerbeke
 formula implies the $W$-constraints of the $\tau$ function. In the end of the work we also
 give a Grassmannian description of the string equation in terms of the spectral parameter.

The paper is organized as follows. In section 2, we recall some basic
notions for the BKP hierarchy such as Lax formulation, (adjoint-) wave function,
$\tau$ function and (differential) Fay identity.
In section 3, we introduce the Orlov-Schulman
operator for additional symmetries and vertex operators for
B\"acklund symmetries of the BKP hierarchy. We give a simple expression of the generators
of vertex operator  using the Fa\`a di Bruno polynomials and provide the
Adler-Shiota-van Moerbeke formula that connects additional symmetries acting on wave function
with those Sato's B\"acklund symmetries acting on tau function. In section 4, we show that
the solution of the $p$-reduced BKP hierarchy constrained by the string equation
can be characterized by a vacuum condition so that the associated tau function is
annihilated by a set of differential operators. In section 5, We show that
these differential operators constitute a $W_p^{B+}$ algebra which contains a
Virasoro algebra as a subalgebra with central charge $c=p$.
In section 6, we provide the Grassmannian description of the string equation in terms of
the spectral parameter. Section 7 is devoted to the concluding remarks.

\section{BKP hierarchy}
In this section we recall some basic properties for the BKP hierarchy \cite{DJKM} .
We shall follow the notations used in the previous work\cite{Tu07}.
The BKP hierarchycan be formulated in Lax form as
\be
\pa_{2n+1}L=[B_{2n+1},L],\quad B_{2n+1}=(L^{2n+1})_+,\quad n=0,1,2,\cdots
\label{laxeq}
\ee
where the Lax operator is defined by
\be
L=\pa+u_1\pa^{-1}+u_2\pa^{-2}+\cdots,
\label{laxop}
\ee
with coefficient functions $u_i$ depending on the time variables $t=(t_1=x,t_3,t_5,\cdots)$
and satisfies the constraint
\be
L^*=-\pa L\pa^{-1}.
\label{constr}
\ee
Here and the rest of the paper we will use the notations:
$(\sum_ia_i\pa^i)_+=\sum_{i\geq 0}a_i\pa^i$,
$(\sum_ia_i\pa^i)_-=\sum_{i< 0}a_i\pa^i$, $(\sum_ia_i\pa^i)_{[k]}=a_k$,
$\resi(\sum_ia_i\pa^i)=a_{-1}$ and
$(\sum_ia_i\pa^i)^*=\sum_i(-\pa)^ia_i$.
It can be shown \cite{DJKM} that the constraint (\ref{constr}) is
equivalent to the condition $(B_{2n+1})_{[0]}=0$.

The Lax equation (\ref{laxeq}) is equivalent to the compatibility condition of the
linear system
\be
Lw(t,z)=z w(t,z),\qquad \pa_{2n+1}w(t,z)=B_{2n+1}w(t,z),
\label{lineq}
\ee
where $w(t,z)$ is called wave function (or Baker function) of the system and $z$
is the spectral parameter.
The whole hierarchy can be expressed in terms of a dressing operator, the so-called
Sato's operator $W$, so that
\[
L=W\pa W^{-1},\qquad W=1+\sum_{j=1}w_j\pa^{-j},
\]
and the Lax equation is equivalent to the Sato's equation
\be
\pa_{2n+1}W=-(L^{2n+1})_-W,
\label{satoeq}
\ee
with constraint\cite{S89,T93}
\be
W^*\pa W=\pa.
\label{constrsato}
\ee
Let the solutions of the linear system (\ref{lineq}) be the form
\be
w(t,z)=We^{\xi(t,z)}=\hat{w}(t,z)e^{\xi(t,z)},
\label{wavefun}
\ee
where $\xi(t,z)=\sum_{i=0}t_{2i+1}z^{2i+1}$ and
$\hat{w}(t,z)=1+w_1/z+w_2/z^2+\cdots$.
Then $w(t,z)$ is a wave function of the BKP hierarchy
if and only if it satisfies the bilinear identity\cite{DJKM}
\be
\resi_z(z^{-1}w(t,z)w(t',-z))=1,\quad \forall t, t'
\label{bileq}
\ee
where we denote the symbol $\resi_z(\sum_ia_iz^i)=a_{-1}$.
In fact, from the bilinear identity (\ref{bileq}), solutions of the BKP hierarchy
can be characterized by a single function $\tau(t)$ called $\tau$-function
such that\cite{DJKM}
\be
\hat{w}(t,z)=\frac{\tau
(t_1-\frac{2}{z},t_3-\frac{2}{3z^3},t_5-\frac{2}{5z^5},\cdots)}{\tau(t)}.
\label{tau}
\ee
From (\ref{wavefun}) and (\ref{tau}) the wave function $w(t,\la)$ can be expressed
in terms of $\tau$-function as
\[
w(t,z)=\frac{X_B(t,z)\tau(t)}{\tau(t)},
\]
where $X_B(t,z)$ is the so-called vertex operator, defined by \cite{DJKM}
\[
X_B(t,z)=e^{\xi(t,z)}e^{-2D(t,z)}\equiv e^{\xi(t,z)}G(z),
\]
with $D(t,z)=\sum_{n=0}z^{-2n-1}\pa_{2n+1}/(2n+1)$.

In the following we provide two useful identities associated with the tau function of the BKP
hierarchy.
\begin{proposition}\cite{Tu07} \label{Fay}(Fay identity) The tau function of the BKP hierarchy
satisfies the Fay  quadrisecant  identity:
\bean
&&\sum_{(s_1,s_2,s_3)}\frac{(s_1-s_0)(s_1+s_2)(s_1+s_3)}{(s_1+s_0)(s_1-s_2)(s_1-s_3)}
\tau(t+2[s_2]+2[s_3])\tau(t+2[s_0]+2[s_1])\nonumber\\
&&+\frac{(s_0-s_1)(s_0-s_2)(s_0-s_3)}{(s_0+s_1)(s_0+s_2)(s_0+s_3)}
\tau(t+2[s_0]+2[s_1]+2[s_2]+2[s_3])\tau(t)=0
\label{Fayeq}
\eean
where $(s_1,s_2,s_3)$ stands for cyclic permutations of $s_1$, $s_2$ and $s_3$.
\end{proposition}
\begin{proposition}\cite{Tu07} \label{dFay}(Differential Fay identity) The following equation holds.
\bean
&&\left(\frac{1}{s_2^2}-\frac{1}{s_1^2}\right)
\{\tau(t+2[s_1]+2[s_2])\tau(t+2[s_2])-\tau(t+2[s_1]+2[s_2])\tau(t)\}\nonumber\\
&&=\left(\frac{1}{s_2}+\frac{1}{s_1}\right)
\{\pa\tau(t+2[s_2])\tau(t+2[s_1])-\pa\tau(t+2[s_1])\tau(t+2[s_2])\}\nonumber\\
&&\quad+\left(\frac{1}{s_2}-\frac{1}{s_1}\right)\{\tau(t+2[s_1]+2[s_2])\pa\tau(t)-
\pa\tau(t+2[s_1]+2[s_2])\tau(t)\}.
\label{dFayeq}
\eean
\end{proposition}
\section{Additional symmetries and vertex operators}
Based on the work of Orlov and Schulman \cite{OS86}, the Lax equation can be extended by
introducing the Orlov-Schulman operator $M$ defined by
\[
M=W\Ga W^{-1},\quad \Ga=\sum_{n=0}(2n+1)t_{2n+1}\pa^{2n},
\]
which satisfies
\[
\pa_{2n+1}M=[B_{2n+1},M],\quad [L,M]=1.
\]
Thus the linear system (\ref{lineq}) should be extended to
\[
Lw=zw,\quad Mw=\pa_zw,\quad \pa_{2n+1}w=B_{2n+1}w.
\]
Note that on the space of wave function $w(t,z)$, $(L,M)$ is anti-isomorphic to $(z,\pa_z)$
since $[z,\pa_z]=-1$. More general, one has
\[
M^mL^lw=z^l\pa_z^mw,\quad L^lM^mw=\pa_z^mz^lw.
\]
In fact, one can define the adjoint wave function
$w^*(t,z)=(W^*)^{-1}e^{-\xi(t,z)}=-z^{-1}w_x(t,-z)$ and
$M^*=(L^*)^{-1}\pa M\pa^{-1}L^*$ where we have use the fact that
$\Ga^*=\Ga$ and (\ref{constrsato}). Then $[L^*,M^*]=[M,L]^*=-1$, and
\[
L^*w^*=zw^*,\quad M^*w^*=-\pa_zw^*,\quad \pa_{2n+1}w^*=-B_{2n+1}^*w^*.
\]
Motivated by the KP hierarchy, one can introduce a new set of parameters $\hat{t}_{ml}$
so that additional symmetries of the BKP hierarchy can be expressed as
\be
\hat{\pa}_{ml}W=-(A_{ml}(L,M))_-W,
\label{addsato}
\ee
where the generator $A_{ml}(L,M)$, due to (\ref{constrsato}), has the form \cite{L95,Tu07}
\be
A_{ml}(L,M)=M^mL^{l}-(-1)^{l}L^{l-1}M^mL.
\label{addgen}
\ee
Let us introduce another generator $Y_B(\la,\mu)$ of additional symmetries as\cite{L95}
\bea
\label{Ao}
Y_B(\la,\mu)&=&\sum_{m=0}^\infty\frac{(\mu-\la)^m}{m!}
\sum_{l=-\infty}^\infty\la^{-l-m-1}(A_{m,m+l}(L,M))_-,\\
&=&\sum_{m=0}^\infty\frac{(\mu-\la)^m}{m!}
\sum_{l=-\infty}^\infty\la^{-l-m-1}(M^mL^{m+l}-(-1)^{m+l}L^{m+l-1}M^mL)_-.\nonumber
\eea
We mention that for $\la=\mu$, the generator $Y_B(\la,\la)=2\sum_{l=odd}\la^{-l-1}L^l_-$
corresponds to the resolvent operator of the BKP hierarchy.

On the other hand, recalling the vertex operator $X_B(\la,\mu)$ defined by
\[
X_B(\la,\mu)=e^{-\xi(t,\la)}e^{\xi(t,\mu)}G(-\la)G(\mu)
\]
which provides the infinitesimal Sato's B\"acklund transformations\cite{DJKM}
on the space of tau function, namely, if $\tau(t)$ is a solution then
$\tau(t)+\epsilon X_B(\la,\mu)\tau(t)$ is a solution as well.
In fact, one can Taylor expand the vertex operator $X_B(\la,\mu)$ around $\mu=\la$
for large $\la$ as
\[
X_B(\la,\mu)=\sum_{m=0}\frac{(\mu-\la)^m}{m!}W^{(m)}(\la)
=\sum_{m=0}\frac{(\mu-\la)^m}{m!}\sum_{l=-\infty}^\infty\la^{-m-l}W_l^{(m)},
\]
where $W^{(m)}(\la)=\pa_\mu^mX_B(\la,\mu)|_{\mu=\la}$. Introducing the symbol $\alpha(z)=\sum_{n=odd}\alpha_nz^{-n}/n$ with
\[
\alpha_n=
\left\{
\ba{ll}
2\pa/\pa t_n & n\in \mathbb{Z}_{odd}^+ \\
|n|t_{|n|}& n\in \mathbb{Z}_{odd}^-
\ea
\right.
\]
where $\mathbb{Z}_{odd}=\mathbb{Z}_{odd}^+\oplus\mathbb{Z}_{odd}^-$.
Then $\alpha_n$'s satisfy the commutation relations
\[
[\alpha_n,\alpha_m]=2n\de_{n,-m},\quad n\in \mathbb{Z}_{odd}.
\]
The vertex operator $X_B(\la,\mu)$ can be expressed as
\[
X_B(\la,\mu)=:e^{\alpha(\la)-\alpha(\mu)}:
\]
where the normal ordering $::$ demands that $\alpha_{n>0}$ must be placed to the
right of $\alpha_{n<0}$. Therefore,
\[
W^{(m)}(\la)=\pa_\mu^mX_B(\la,\mu)|_{\mu=\la}=:\pa_\la^me^{-\alpha(\la)}\cdot e^{\alpha(\la)}:
=:F_m(-\pa_\la\alpha(\la)):
\]
where $F_m(u(z))$ is the so-called Fa\`a di Bruno polynomials (see e.g. \cite{D03}) defined by
the recurrence relations $F_{m+1}(u)=(\pa_z+u)F_m(u)$. For instance,
$F_0=1, F_1=u, F_2=u'+u^2, F_3=u''+3uu'+u^3$.
\begin{lemma} \cite{D95} The following formula
\be
W^{(m)}(\la)=\sum_{m_1+2m_2+\cdots=m}\frac{m!}{m_1!m_2!\cdots}
:\left(-\frac{\pa_\la\alpha}{1!}\right)^{m_1}\left(-\frac{\pa_\la^2\alpha}{2!}\right)^{m_2}\cdots:
\ee
holds, where $m_i\geq 0$.
\end{lemma}
\begin{proof} Introducing a generating function of $F_m$ as $g(\la,z)=\sum_{m=0}F_mz^m/m!$,
then
\bean
g(\la,z)&=&\sum_{m=0}\frac{(d_\la^m1)}{m!}z^m,\quad
d_\la=e^{\alpha(\la)}\cdot\pa_\la\cdot e^{-\alpha(\la)}\\
&=&e^{\alpha(\la)}\cdot e^{z\pa_\la} e^{-\alpha(\la)}\\
&=&e^{\alpha(\la)} e^{-\alpha(\la+z)}\\
&=&\prod_{n=1}e^{-\pa_\la^n\alpha z^n/n!}\\
&=&\sum_{m_1=0}\frac{\left(-\frac{\pa_\la\alpha}{1!}\right)^{m_1}z^{m_1}}{m_1!}
\sum_{m_2=0}\frac{\left(-\frac{\pa_\la^2\alpha}{2!}\right)^{m_2}z^{2m_2}}{m_2!}\cdots\\
&=&\sum_{m=0}\left(\sum_{m_1+2m_2+\cdots=m}m!
\frac{\left(-\frac{\pa_\la\alpha}{1!}\right)^{m_1}}{m_1!}
\frac{\left(-\frac{\pa_\la^2\alpha}{2!}\right)^{m_2}}{m_2!}\cdots\right)\frac{z^m}{m!}.
\eean
\end{proof}
The vertex operator generators $W_l^{(m)}$ can be easily computed as
\bean
W_n^{(0)}&=&\de_{n,0},\\
W_n^{(1)}&=&
\left\{
\ba{ll}
\alpha_n & n\in {\mathbb Z\/}_{odd} \\
0& n\in {\mathbb Z\/}_{even}
\ea
\right.,\\
W_n^{(2)}&=&
\left\{
\ba{ll}
-(n+1)\alpha_n& n\in {\mathbb Z\/}_{odd}  \\
\sum_{i+j=n}:\alpha_i\alpha_j: & n\in {\mathbb Z\/}_{even}
\ea
\right.,\\
W_n^{(3)}&=&
\left\{
\ba{ll}
\sum_{i+j+k=n}:\alpha_i\alpha_j\alpha_k:+(n+1)(n+2)\alpha_n& n\in {\mathbb Z\/}_{odd} \\
-\frac{3}{2}(n+2)\sum_{i+j=n}:\alpha_i\alpha_j: & n\in {\mathbb Z\/}_{even}
\ea
\right.,\\
W_n^{(4)}&=&
\left\{
\ba{ll}
-2(n+3)\sum_{i+j+k=n}:\alpha_i\alpha_j\alpha_k:-(n+1)(n+2)(n+3)\alpha_n& n\in {\mathbb Z\/}_{odd} \\
\sum_{i+j+k+l=n}:\alpha_i\alpha_j\alpha_k\alpha_l:-\sum_{i+j=n}ij:\alpha_i\alpha_j: & \\
-(2n^2+9n+11)\sum_{i+j=n}:\alpha_i\alpha_j: & n\in {\mathbb Z\/}_{even}
\ea
\right.,
\eean
etc.

A remarkable formula described below provides a bridge between additional symmetries acting
on wave function and Sato's B\"acklund symmetries acting on $\tau$ function.
This kind of formula was first derived for the KP
hierarchy by Adler-Shiota-van Moerbeke\cite{ASM94,ASM95} and Dickey\cite{D95},
and later for BKP by van de Leur\cite{L95} (see also \cite{Tu07}).
\begin{theorem} \cite{L95,Tu07}  The following formula
\be
X_B(\la,\mu)w(t,z)=2\la\left(\frac{\la-\mu}{\la+\mu}\right)Y_B(\la,\mu)w(t,z),
\label{ASM}
\ee
holds for the BKP hierarchy, where it should be understood that the vertex operator
 $X_B(\la,\mu)$ acting on $w(t,z)$ is generated by its action on the $\tau$ function.
\end{theorem}
We remark that the proof of (\ref{ASM}) in \cite{Tu07} is based on a simple expression for the generator
$Y_B(\la,\mu)$ and the differential Fay identity of the BKP hierarchy.
Through the fermion-boson correspondence in the BKP hierarchy, a realization of Lie algebra
$go(\infty)$ on ${\mathbb C\/}[t_1,t_3,t_5,\cdots]$ is given by\cite{DJKM}
\[
Z_B(\la,\mu)=\frac{1}{2}\frac{\mu+\la}{\mu-\la}(X_B(\la,\mu)-1),
\label{ZB}
\]
which after Taylor expanding around $\mu=\la$ for large $\la$ has the form
\[
Z_B(\la,\mu)=\sum_{m=0}\frac{(\mu-\la)^m}{m!}\sum_{l=-\infty}^\infty\la^{-m-l}Z_l^{(m+1)}.
\]
It is easy to show that differential operators $Z_l^{(m)}$ are related to $W_l^{(m)}$ as
\[
Z_l^{(1)}=W_l^{(1)},\quad
Z_l^{(m+1)}=\frac{W_l^{(m+1)}}{m+1}+\frac{1}{2}W_l^{(m)},\quad m\geq 1,
\]
and constitute  an infinite-dimensional Lie algebra called
 $W^B_{1+\infty}$-algebra which is a subalgebra of $W_{1+\infty}$ associated
 with the KP hierarchy.
\section{$p$-reduced BKP and string equation}
The so-called $p$-reduced BKP hierarchy\cite{DJKM} is defined by the Lax operator
(\ref{laxop}) such that
\be
L^p=(L^p)_+,\quad p=\mbox{odd integer}.
\label{pred}
\ee
This reduction must be compatible with the condition (\ref{constr}).
For example, $(L^p)_{[0]}=0$ for odd $p$ and thus
\[
(L^p)^*_+=(L^*)^p=(-1)^p\pa L^p\pa^{-1}=(-1)^p\pa (L^p)_+\pa^{-1}
\]
is a differential operator as well.
Therefore, from the Lax equation (\ref{laxeq}), we have
\be
\pa_{2n+1}L^p=\left[(L^p)^{\frac{2n+1}{p}}_+,L^p\right]
\label{plaxeq}
\ee
which implies that $L^p$ is independent of the parameters $t_{jp}$ for $j=1,3,5,\cdots$.
For $p=1$ case, we have $L=L_+=\pa$ which is a trivial reduction ($u_i=0$ for $\forall i$).
The next case is $p=3$, which contains the simplest nontrivial equation called
Sawada-Kotera equation\cite{SK74}:
\[
u_t+15(u^3+uu_{xx})_x+u_{xxxxx}=0
\]
where $x=t_1$, $t=t_5$ and $u=2(\log\tau)_{xx}$ with $\pa\tau/\pa t_{3j}=0\; (j=1,3,5,\cdots)$.

In the following , we would like to characterize the solutions of the $p$-reduced BKP
hierarchy constrained by the string equation
\be
[L^p,P]=1
\label{streq}
\ee
where $P$ is a differential operator. From (\ref{addgen}), we have
\be
 A_{1,1-k} =  \left\{
            \ba{ll}
              kL^{-k}  & k \in \mathbb{Z}_{odd} \\
             2ML^{1-k}-kL^{-k} &k \in \mathbb{Z}_{even}
            \ea  \right.
            \label{add-1}
\ee
Guided by the KP hierarchy\cite{D93}, if we think the string equation (\ref{streq}) as
the consequence of the addition flow equation $\hat{\pa}_{1,1-p}W=0$ for odd $p$,
then (\ref{addsato}) and (\ref{add-1}) show that $(A_{1,1-p})_-=pL^{-p}=0$ which
produces a contradictory result.
In\cite{L96} van de Luer pointed out that one may consider the additional flow equation
$\hat{\pa}_{1,1-2p}W=-(A_{1,1-2p})_-W=0$ from which one gets
\be
(ML^{1-2p})_-=pL^{-2p}
\label{1-2p}
\ee
and hence the operator $Q\equiv (ML^{1-2p}-pL^{-2p})/2p$ is purely differential. Then
\be
[L^{2p}, Q]=1.
\label{pre-str}
\ee
However, in view of the fact that $L^p=(L^p)_+$, we have
\be
(ML^{1-p})_-=(ML^{1-2p}L^p)_-=pL^{-p},
\label{1-p}
\ee
which provides the differential operator $P\equiv(ML^{1-p}-pL^{-p})/p$
for the string equation (\ref{streq}). Therefore, (\ref{1-2p}) can be regarded as the
symmetry origin of the string equation (\ref{streq}), and thus we refer (\ref{1-2p})
to the pre-string equation.
Taking the residue of (\ref{1-2p}) we obtain
\be
\sum_{n=p}(2n+1)t_{2n+1}\resi(L^{2n-2p+1})+(2p-1)t_{2p-1}=0.
\label{strres}
\ee
From the Sato's equation (\ref{satoeq}) and the formula (\ref{tau}), we have
\[
\resi(L^{2n-2p+1})=2\pa_1\pa_{2n-2p+1}\log \tau.
\]
Substituting above back to (\ref{strres}) and integrating it over $t_1$,
yields
\be
\left(\sum_{n=p}(2n+1)t_{2n+1}\frac{\pa}{\pa t_{2n+1-2p}}+
\frac{(2p-1)}{2}t_1t_{2p-1}+c\right)\tau(t)=0.
\label{L-1}
\ee
  In fact, from (\ref{1-2p}) and the $p$-reduced BKP hierarchy flows,
  we can prove a more general result of constraints on $\tau$-function.
\begin{proposition}
For the $p$-reduced Lax operator $L^p$ constrained by the string equation, the following
formulas hold for $m\geq 0$.
\bea
\label{F1}
 (M^mL^{jp+m})_{-}& =&  \left\{
   \ba{ll}
    \prod^{m-1}_{r=0} (p-r)L^{-2p} & j = -2 \\
   \prod^{m-1}_{r=0} (p-r)L^{-p} & j = -1 \\
     0 & j = 0,1,2,...
   \ea  \right.\\
   \label{F2}
 (L^{jp+m-1}M^mL)_{-}& =&(-1)^m(M^mL^{jp+m})_{-}
 \eea
 with the proviso that the factor $\prod^{m-1}_{r=0}(p-r)$ should be set to 1 for $m=0$.
\end{proposition}
\begin{proof}
We shall follow \cite{AM92} to prove it by induction on $m$ and $j$.
For $m=0$, the proof is obvious. For $m=1,j=-2$, this is just the
equation (\ref{1-2p}). Assume that the formula (\ref{F1})  holds up to
some integer $m>0$ for $j=-2$, then for $j\geq -1$, one sees that
 \bean
   (M^mL^{jp+m})_{-} &=& ((M^mL^{m-2p})_{-}L^{(j+2)p})_{-}\\
                      &=& \left(\prod^{m-1}_{r=0} (p-r)L^{-2p}L^{(j+2)p}\right)_{-}  \\
                      &=&  \left\{
   \ba{ll}
      \prod^{m-1}_{r=0} (p-r)L^{-p} & j = -1 \\
     0 & j = 0,1,2,...
   \ea  \right..
 \eean
 Next, for the $j=-2$ , we have
\bean
   (M^{m+1}L^{m+1-2p})_{-} &=&(M^mML^mL^{1-2p})_{-}\\
                           &=&(M^mL^mML^{1-2p})_{-}-m(M^mL^{m-2p})_{-}\\
                           &=&(M^mL^m(ML^{1-2p})_{-})_{-}-m(M^mL^{m-2p})_{-}\\
                           &=&p(M^mL^{m-2p})_{-}-m(M^mL^{m-2p})_{-}  \\
                           &=&(p-m)\prod^{m-1}_{r=0} (p-r)L^{-2p}\\
                           &=&\prod^{m}_{r=0} (p-r)L^{-2p}
 \eean
where the third equality is due to the fact that $M^mL^m$ is a differential operator.
The formula (\ref{F2}) can be proved in a similar way.
\end{proof}
We like to mention that  a similar result as (\ref{F1}) for the CKP hierarchy has been derived
in \cite{HTFM07} to discuss the associated additional symmetries and string equation.
\begin{proposition} Let $L^p$ be the Lax operator of the $p$-reduced BKP hierarchy constrained by
the string equation (\ref{streq}). Then for $m\geq 0$ and $j\geq -2$,
\be
Z^{(m+1)}_{jp}\tau(t)=c\cdot\tau(t).
\label{wc}
 \ee
 where c is a constant.
 \end{proposition}
\begin{proof}
Using (\ref{Ao}) and the Adler-Shiota-van Moerbeke formula (\ref{ASM}),
 we have
\bean
(A_{m,m+jp})_-\om(t,z)
&=&\resi_\la\left(\la^{jp+m}\pa_\mu^m|_{\mu=\la}Y_B(\la,\mu)\om(t,z)\right)\\
&=&\resi_\la\left(\la^{jp+m-1}\pa_\mu^m|_{\mu=\la}
\frac{1}{2}\frac{\la+\mu}{\la-\mu}X_B(\la,\mu)\om(t,z)\right)\\
&=&-\om(t,z)(G(z)-1)\resi_\la\left(\la^{jp+m-1}\pa_\mu^m|_{\mu=\la}
\frac{Z_B(\la,\mu)\tau(t)}{\tau(t)}\right)\\
&=&-\om(t,z)(G(z)-1)\left(\frac{Z^{(m+1)}_{jp}\tau(t)}{\tau(t)}\right).
\eean
On the other hand, from (\ref{F1}) and (\ref{F2}) we have
\bean
(A_{m,m+jp})_-\om(t,z)
&=&((M^mL^{jp+m})_--(-1)^{jp+m}(L^{jp+m-1}M^mL)_-)\om(t,z)\\
&=&(1-(-1)^{jp})(M^mL^{jp+m})_-\om(t,z)\\
&=&2\prod^{m-1}_{r=0} (p-r)L^{-p}\om(t,z)\delta_{j,-1}\\
&=& 2\prod^{m-1}_{r=0} (p-r)z^{-p}\om(t,z)\delta_{j,-1}.
\eean
Noticing that
\[
2z^{-p}=-(p(t_p-2z^{-p}/p)-pt_p)=-(G(z)-1)\frac{Z_{-p}^{(1)}\tau(t)}{\tau(t)},
\]
and hence
\[
(G(z)-1)\left(\frac{Z^{(m+1)}_{jp}\tau(t)}{\tau(t)}-
\prod^{m-1}_{r=0} (p-r)\delta_{j,-1}\frac{Z_{-p}^{(1)}\tau(t)}{\tau(t)}\right)=0.
\]
Since  $(G(z)-1)f(t)=0$ implies  $f(t)$ is a constant, we have
\[
\left(Z^{(m+1)}_{jp}-\prod^{m-1}_{r=0} (p-r)\delta_{j,-1}Z_{-p}^{(1)}\right)\tau(t)=c\cdot\tau(t).
\]
Now we can drop the term involving $Z_{-p}^{(1)}=pt_p$ without harm because the $p$-reduced
BKP hierarchy does not depend on the variables $t_p$.
\end{proof}
\section{String equation as the lowest Virasoro constraint}
In this section, we like to discuss the algebraic structure of the equation (\ref{wc}).
Given a infinite-dimensional algebra $W_{1+\infty}^B$ defined by the vertex operator,
 one can introduce two subalgebras as follows:
\[
W^B_p=\{\mbox{generated by }
W^{(m)}_{jp}, 1\leq m \leq p,j\in \mathbb{Z},  t_p=t_{3p}=\cdots=0\}
\]
and its truncated sub-algebra:
\[
 W^{B+}_p=\{\mbox{generated by }W^{(m)}_{jp}, 1\leq m
\leq p,j\geq -2, t_p=t_{3p}=\cdots=0\}
\]
We shall show that the $\tau$ function satisfying the $p$-reduced BKP hierarchy and
the string equation is a null-vector of the $W_p^{B+}$-algebra.
Since the algebra $W_p^+$ for the KP hierarchy has no central
term\cite{FKN92}, thus we expect that the subalgebra $W_p^{B+}\subset W_p^+$
 has also no central term.
 To see this, we have to properly combine the generators $Z^{(m)}_{jp}$ so that
every redefined element $\mathbf{W}_n^{(m)}$ of $W_p^{B+}$ can be expressed as a
commutator of two elements of $W_p^{B+}$.
As a consequence, the constant $c$ in (\ref{wc}) can be removed and
\be
\mathbf{W}_n^{(m)}\tau(t)=0,\quad 1\leq m\leq p,\; n\geq -m+1.
\label{Wc}
\ee
We remark that the condition for the subscript $n$ in (\ref{Wc}) is due to the fact that
 for higher-spin generators $\mathbf{W}^{(m)}_n$, one has
\[
[\mathbf{W}^{(2)}_{-1},\mathbf{W}^{(m)}_n]=(-(m-1)-n)\mathbf{W}^{(m)}_{n-1},
\]
and thus, under bracketing with generators $\mathbf{W}^{(2)}_{-1}$, one can reach
the lowest one $\mathbf{W}^{(m)}_{-m+1}$ which commutes with $\mathbf{W}^{(2)}_{-1}$.
Let us demonstrate the first few $W$-constraints.

For $m=0$, (\ref{wc}) shows that
\[
Z^{(1)}_{jp}\tau(t)=W_{jp}^{(1)}\tau(t)=2\pa\tau(t)/\pa t_{jp}=0,\quad j=1,3,5,\cdots
\]
which is just the condition $L^p=(L^p)_+$ for $p$-reduced BKP hierarchy.
Thus we have
\[
\mathbf{W}_n^{(1)}\tau(t)=W_{(2n+1)p}^{(1)}\tau(t)=0, \quad n\geq 0
\]
For $m=1$, (\ref{wc}) shows that $Z_{2kp}^{(2)}\tau=c^{(2)}_k(p)\tau$ with  $k\geq -1$ where
\[
Z^{(2)}_{2kp}=2\sum_{n=0}^{kp-1}\pa_{2n+1}\pa_{2kp-2n-1}+2\sum_{n=0}^\infty
(2n+1)t_{2n+1}\pa_{2kp+2n+1}
\]
for $k\geq 0$ and
\bean
Z^{(2)}_{2kp}&=&\frac{1}{2}\sum_{n=0}^{-kp-1}(2n+1)(-2kp-2n-1)t_{2n+1}t_{-2kp-2n-1}\\
 &&+ 2\sum_{n=0}(2n+1-2kp)t_{2n+1-2kp}\pa_{2n+1}
\eean
for $k<0$. Define $l_k = Z^{(2)}_{2kp}/4p$ with  $k\geq -1$ then
 $l_k$ satisfy centerless Virasoro algebra $[l_n,l_m]=(n-m)l_{n+m}$ except
that
\[
[l_1,l_{-1}]=2\left(l_0+\frac{1}{16p}+\frac{p^2-1}{24p}\right).
 \]
 This means that the constants $c^{(2)}_{k\neq 0}(p)=0$ and
 $c^{(2)}_0(p)=-(1/4+(p^2-1)/6)$.
If we redefine a new set of operators as
 \[
\mathbf{W}_n^{(2)}\equiv\mathbf{L}_n=l_n+\delta_{n,0}\left(\frac{1}{16p}+
\frac{p^2-1}{24p}\right),\quad n\geq -1
\]
then we have $[\mathbf{L}_n,\mathbf{L}_m]=(n-m)\mathbf{L}_{n+m}$ for $m,n\geq -1$.
Therefore, the Virasoro constraint becomes
\[
\mathbf{W}_n^{(2)}\tau(t)=0,\quad n\geq -1.
\]
In particular, the lowest Virasoro constraint $\mathbf{L}_{-1}\tau=0$ is given by
\[
\left( \frac{1}{2}\sum_{n=0}^{p-1}(2n+1)(2p-2n-1)t_{2n+1}t_{2p-2n-1}\\
 + 2\sum_{n=0}(2n+1+2p)t_{2n+1+2p}\pa_{2n+1}\right )\tau=0
\]
which is just the pre-string equation (\ref{L-1}).
Moreover, when $n$ extends to all integers, $\mathbf{W}_n^{(2)}$ indeed
constitute the generators of a standard Virasoro algebra with central charge $c=p$, namely,
\[
[\mathbf{L}_n,\mathbf{L}_m]=(n-m)\mathbf{L}_{n+m}+\delta_{n+m,0}\frac{n^3-n}{12}p.
\]
For $m=2$, we may define the spin-3 generators as $\mathbf{W}_n^{(3)}=Z^{(3)}_{(2n+1)p}$
which indeed satisfy the commutation relation
\[
[\mathbf{L}_n,\mathbf{W}^{(3)}_m]=(2n-m)\mathbf{W}^{(3)}_{n+m}
\]
and the constraint equation
\[
\mathbf{W}_n^{(3)}\tau(t)=0,\quad n\geq -2.
\]
The higher-spin generators can be treated in a similar manner.
\section{Grassmannian description of the string equation}
In this section we like to give a geometric description of the string equation for
the $p$-reduced BKP hierarchy.
Let $H$ be a Hilbert space defined by  formal power series in $z$
that can be decomposed into two infinite-dimensional
subspaces as $H=H_+\oplus H_-$ where
\[
H_+=\mbox{span}\{ z^0, z^1, z^2,\cdots\},\quad
H_-=\mbox{span}\{ z^{-1}, z^{-2}, z^{-3},\cdots\}.
\]
 The Grassmannian $Gr$ is defined by the set of all subspaces
$V\subset H$ with the following conditions:\cite{Sato84,SW85}
\[
Gr=\{V|V\subset H, p_+|_V: V\to H_+ \mbox{ (Fredholm)},
p_-|_V: V\to H_-\mbox{ (compact)}\}
\]
where $p_\pm$  are projection operators.
 If $p_+|_V: V\to H_+$ is a bijection, then $V$ is called transversal to $H_-$,
 or transversal for short (i.e. $V$ belongs to the big cell $Gr_0\subset Gr$).
 The KP hierarchy can be regarded as a simple dynamical system on $Gr$
 (see \cite{Sato84} for the detail).

The BKP hierarchy is defined by the subvariety $Gr^B\subset Gr$(see e.g. \cite{S89}) so that
  an infinite-dimensional plane $V^0\in Gr^B$ can  be
 represented as follows\cite{Sato84,SW85}:
 \bean
V^0&=&\mbox{span}\{\om(t,z)|_{t=0}, \pa_x\om(t,z)|_{t=0},\pa_x^2\om(t,z)|_{t=0},\cdots\}\\
&=&\mbox{span}\{\om(t,z), \mbox{for all}\;\; t\in \mathbb{C}^\infty \}.
 \eean
 where $t=(t_1,t_3,t_5,\cdots)$ and $\om(t,z)$ satisfies the bilinear equation (\ref{bileq}).
For the $p$-reduced BKP hierarchy, $\pa_pW=-(L^p)_-W$ and
the wave function $\om(t,z)$ associated with $V^0$ satisfies
\[
\pa_p\om(t,z)=(L^p)_+\om(t,z)=L^p\om(t,z)=z^p\om(t,z)\in V^0
\]
Therefore
\[
z^pV^0\subset V^0.
\]
On the other hand, the  string equation (\ref{streq}) can be traced back to the Sato
equation of the additional flow $\hat{\pa}_{1,1-2p}W=-(A_{1,1-2p})_-W=0$ which follows
\[
\hat{\pa}_{1,1-2p}\om(t,z)=-(A_{1,1-2p})_-\om(t,z)=0
\]
and hence
\bean
Q\om(t,z)&=&\frac{1}{2p}(ML^{1-2p}-pL^{-2p})\om(t,z)\\
&=&\frac{1}{2p}(z^{1-2p}\pa_z-pz^{-2p})\om(t,z)\\
&\equiv&A_{2p}\om(t,z)\in V^0.
\eean
Therefore, the plane $V^0\in Gr^B$ associated with $\omega(t,z)$ of $p$-reduced BKP hierarchy
constrained by the string equation $[L^p,P]=1$ is invariant
under the action of differential operators $L^p$ and
$Q\equiv\frac{1}{2p}(ML^{1-2p}-pL^{-2p})$. They act as  $z$-operators:
\[
L^p\mapsto z^p,Q\mapsto A_{2p}\equiv z^{p}\frac{d}{dz^{2p}}z^{-p}
\]
and
\be
z^p V^0\subset V^0 ,A_{2p}V^0 \subset V^0, \mbox{with}\;\;
[A_{2p},z^{2p}]=1.
\label{z-str}
\ee
The above discussions enable us to transform the original problem for solving
the solution of the $p$-reduced BKP hierarchy constrained by the string equation
to that described in $z$-operators in the context of  Grassmannian.

\section{Concluding Remarks}
We have investigated the string equation of the BKP hierarchy from additional
symmetries point of view. We show that the $p$-reduced BKP hierarchy constrained
by the string equation can be formulated in terms of Lax and Orlov-Schulman operators
as what has been done for the KP hierarchy. In particular, the invariance of the
additional symmetry with respect to the $\hat{t}_{1,1-2p}$-flow is crucial for
obtaining the string equation which together with the  Adler-Shiota-van Moerbeke
 formula implies the $W$-constraints of the $\tau$ function.
Furthermore, we show that the Lax-Orlov-Schulman formulation of the
$p$-reduced BKP hierarchy with string equation can be transformed
to that of Sato Grassmannian in terms of the spectral parameter $z$.
In view of the works by Kac and Schwarz\cite{KS91,S91},  the Grassmannian
description might provide a starting point for investigating the existence problem
of BKP solutions characterized by the string equation.
 We hope to address this issue in our future work.

{\bf Acknowledgments\/}\\
We like to thank N.C. Lee for helpful discussions.
This work is supported by the National Science Council of Taiwan
under Grant  NSC97-2112-M-194-002-MY3.


\end{document}